\documentclass[reqno]{amsart}

\usepackage{amsmath, amssymb}
\usepackage[mathscr]{eucal}
\usepackage{upgreek}
\usepackage{hyperref}
\hypersetup{colorlinks=true, linkcolor=blue, urlcolor=blue, citecolor=[rgb]{0, 0.8, 0}}
\usepackage{tikz}
\usetikzlibrary{arrows.meta,positioning}

\theoremstyle{plain}
\newtheorem{theorem}{Theorem}[section]
\newtheorem{lemma}[theorem]{Lemma}

\theoremstyle{remark}
\newtheorem{remark}[theorem]{Remark}

\numberwithin{equation}{section}

\newcommand{\du}{\mathrm{d}}

\DeclareMathOperator{\ind}{ind}
\DeclareMathOperator{\Trace}{\mathbf{Trace}}

\title[Singularities and eigenparameter dependent boundary conditions]{Inverse square singularities and \\ eigenparameter dependent boundary conditions \\ are two sides of the same coin}

\author{Namig J. Guliyev}
\address{Institute of Mathematics and Mechanics, Azerbaijan National Academy of Sciences, 9 B.~Vahabzadeh str., AZ1141, Baku, Azerbaijan.}
\email{njguliyev@gmail.com}



\begin{document}
\begin{abstract}
We show that inverse square singularities can be treated as boundary conditions containing rational Herglotz--Nevanlinna functions of the eigenvalue parameter with ``a negative number of poles''. More precisely, we treat in a unified manner one-dimensional Schr\"{o}dinger operators with either an inverse square singularity or a boundary condition containing a rational Herglotz--Nevanlinna function of the eigenvalue parameter at each endpoint, and define Darboux-type transformations between such operators. These transformations allow one, in particular, to transfer almost any spectral result from boundary value problems with eigenparameter dependent boundary conditions to those with inverse square singularities, and vice versa.
\end{abstract}
\maketitle

\tableofcontents

\section{Introduction} \label{sec:introduction}

One-dimensional Schr\"{o}dinger operators with summable potentials and constant (Dirichlet or Robin) boundary conditions are among the most comprehensively studied operators in spectral theory \cite{M77}, \cite{PT87}. However, potentials arising in physical applications usually have nonintegrable singularities of various kind. For instance, the potential of the Schr\"{o}dinger equation corresponding to one of the most basic models of quantum mechanics, namely that of the nonrelativistic hydrogen atom, is represented as the sum of two nonintegrable terms, one of which is a multiple of $x^{-1}$, while the other is of the form $\ell (\ell + 1) x^{-2}$, where the orbital angular momentum quantum number $\ell$ is a nonnegative integer~\cite[Section 39]{D67}. Potentials with singularities of the former kind can be studied just like summable ones~\cite{G19a} by introducing the so-called quasi-derivatives~\cite{AEZ88}. Inverse square singularities (also known as Bessel-type singularities) of the above form, however, also arise from other radial Schr\"{o}dinger operators acting in odd-dimensional spaces by separation of variables, and require a special treatment. We refer the reader to~\cite{KST10} and the references therein for various direct and inverse spectral results concerning inverse square singularities.

On the other hand, eigenvalue problems with boundary conditions dependent on the eigenvalue parameter have found numerous applications in various fields of science and technology. Fulton's well-known papers \cite{F77}, \cite{F80} contain a fairly comprehensive list of references up to 1980 and some recent examples of their use in modern physics can be found in~\cite{G19c}. We refer to our recent papers \cite{G17}, \cite{G19a}, and the references therein for direct and inverse spectral theory of one-dimensional Schr\"{o}dinger operators with eigenparameter dependent boundary conditions.

From an operator-theoretic point of view, a boundary value problem with boundary conditions containing (not necessarily rational) Herglotz--Nevanlinna functions of the eigenvalue parameter describes a \emph{self-adjoint extension with exit} of the minimal differential operator generated by the one-dimensional Schr\"{o}dinger equation in $\mathscr{L}_2$, i.e. a self-adjoint operator in a larger Hilbert space $\mathfrak{H}$ containing $\mathscr{L}_2$ as a closed subspace~\cite{G18b}. Problems with boundary conditions containing rational Herglotz--Nevanlinna functions then correspond to extensions with finite-dimensional \emph{exit spaces} (i.e., $\dim \mathfrak{H} \ominus \mathscr{L}_2 < \infty$). Such extensions have recently been studied in \cite{DL17}, \cite{DL18}.

In our recent paper \cite{G17}, we defined Darboux-type transformations between boundary value problems with boundary conditions containing rational Herglotz--Nevanlinna functions of the eigenvalue parameter, and obtained various direct and inverse spectral results for such problems by using these transformations. For problems with a boundary condition dependent on the eigenvalue parameter (respectively, an inverse square singularity) at one endpoint and a Robin or Dirichlet condition at the other, a similar approach was used earlier by Binding, Browne and Watson \cite{BBW02a}, \cite{BBW02b} (respectively, by Carlson~\cite{C93}).

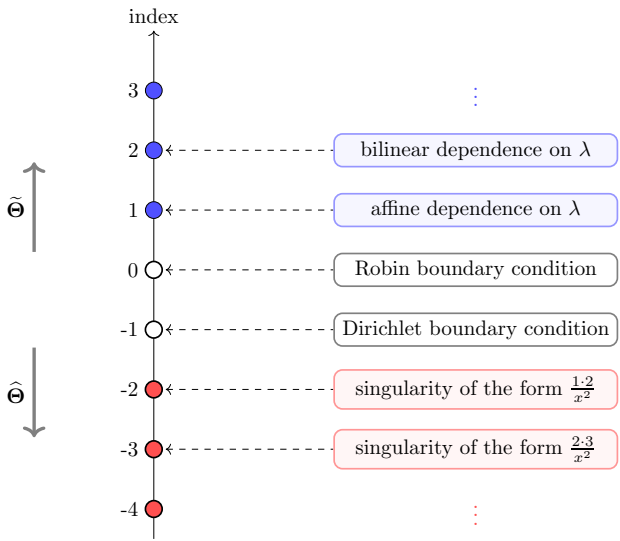
\begin{figure}[ht]
  \resizebox{0.65\textwidth}{!}{\begin{tikzpicture}
    \definecolor{singular}{RGB}{255,80,80}
    \definecolor{constant}{RGB}{255,255,255}
    \definecolor{dependent}{RGB}{80,80,255}
    \pgfmathsetmacro\dotradius{0.14}

    \draw[ultra thick, gray, ->] (-2,0.3) -- node[black, left] {$\widetilde{\boldsymbol{\Theta}}$}(-2,1.8);

    \draw[ultra thick, gray, ->] (-2,-1.3) -- node[black, left] {$\widehat{\boldsymbol{\Theta}}$}(-2,-2.8);

    \draw[->] (0,-4.5) -- (0,4) node[above] {index};

    \filldraw[fill=dependent, draw=black] (0,3) circle (\dotradius) node[left, inner sep=7pt] {3};
    \filldraw[fill=dependent, draw=black] (0,2) circle (\dotradius) node[left, inner sep=7pt] {2};
    \filldraw[fill=dependent, draw=black] (0,1) circle (\dotradius) node[left, inner sep=7pt] {1};
    \filldraw[fill=constant, draw=black, thick] (0,0) circle (\dotradius) node[left, inner sep=7pt] {0};
    \filldraw[fill=constant, draw=black, thick] (0,-1) circle (\dotradius) node[left, inner sep=7pt] {-1};
    \filldraw[fill=singular, draw=black, thick] (0,-2) circle (\dotradius) node[left, inner sep=7pt] {-2};
    \filldraw[fill=singular, draw=black, thick] (0,-3) circle (\dotradius) node[left, inner sep=7pt] {-3};
    \filldraw[fill=singular, draw=black, thick] (0,-4) circle (\dotradius) node[left, inner sep=7pt] {-4};

    \tikzset{
      every node/.style={rectangle, rounded corners, solid, right, align=center, text width=4.5cm, thick},
      singularnode/.style={draw=singular!60, fill=singular!5},
      constantnode/.style={draw=gray},
      dependentnode/.style={draw=dependent!60, fill=dependent!5},
    }

    \node[dependent] at (3,3) {$\vdots$};
    \draw[dashed, <-] (0.2,2) to [out=0,in=180] (3,2) node[dependentnode] {bilinear dependence on $\lambda$};
    \draw[dashed, <-] (0.2,1) to [out=0,in=180] (3,1) node[dependentnode] {affine dependence on $\lambda$};
    \draw[dashed, <-] (0.2,0) to [out=0,in=180] (3,0) node[constantnode] {Robin boundary condition};
    \draw[dashed, <-] (0.2,-1) to [out=0,in=180] (3,-1) node[constantnode] {Dirichlet boundary condition};
    \draw[dashed, <-] (0.2,-2) to [out=0,in=180] (3,-2) node[singularnode] {singularity of the form $\frac{1 \cdot 2}{x^2}$};
    \draw[dashed, <-] (0.2,-3) to [out=0,in=180] (3,-3) node[singularnode] {singularity of the form $\frac{2 \cdot 3}{x^2}$};
    \node[singular] at (3,-4) {$\vdots$};
  \end{tikzpicture}}
  \caption{Singularities, boundary conditions, and their indices}
  \label{fig:index}
\end{figure}

The purpose of this paper is to demonstrate that the framework developed in \cite{G17} can be naturally extended to include inverse square singularities. This is by no means unexpected, since the fact that constant boundary conditions are related with eigenparameter dependent boundary conditions and inverse square singularities via Darboux transformations was already observed by Churchill~\cite{C42} and Crum~\cite{C55} (see also Krein~\cite{K57}) respectively. In~\cite{G17}, to each eigenparameter dependent or independent boundary condition we assigned its \emph{index}, an integer not smaller than $-1$, and two other real numbers, which allowed us to formulate various spectral results in a unified manner. What is perhaps more surprising is that after assigning the index with values smaller than $-1$ to inverse square singularities (see Figure~\ref{fig:index}) and defining the other two numbers in a proper way (see $\omega_1$ and $\omega_2$ in (\ref{eq:omega_Omega}) below), the results obtained in~\cite{G17} either remain intact in our more general setting or can easily be generalized (see Subsection~\ref{ss:trace} for the latter case). Moreover, the index of the inverse square singularity with orbital quantum number $\ell$ under this assignment simply turns out to be equal to $-\ell - 1$, the only value of $\kappa$ other than $\ell$ for which $\kappa (\kappa + 1) = \ell (\ell + 1)$. On the other hand, the index of a rational Herglotz--Nevanlinna function counts the number of its poles: each finite pole is counted twice and a pole at infinity (if any) once. It is in this sense that inverse square singularities behave as boundary conditions containing rational Herglotz--Nevanlinna functions of the eigenvalue parameter with ``a negative number of poles''.

The (Darboux) transformation method used in this paper is also known as the \emph{single commutation method} in contradistinction to the \emph{double commutation method}~\cite{GT96}, which was also used for studying problems with an inverse square singularity at an endpoint \cite{G65} (see also~\cite{AHM07}). In terms of their effect on the spectral measure---which in our case is just the sum $\sum \gamma_n^{-1} \delta_{\lambda_n}(t)$ of the Dirac measures supported at the eigenvalues, with the \emph{norming constants} $\gamma_n$ being defined in Subsection~\ref{ss:hilbert}---the former method amounts to multiplying this measure by the factor $t - \lambda$ (cf. Theorem~\ref{thm:transformation}), whereas the latter one to adding a weighted Dirac measure supported at $\lambda$ to it~\cite{KST12}. The double commutation method, however, appears to be not well suited for studying problems with eigenparameter dependent boundary conditions or inverse square singularities at both endpoints.

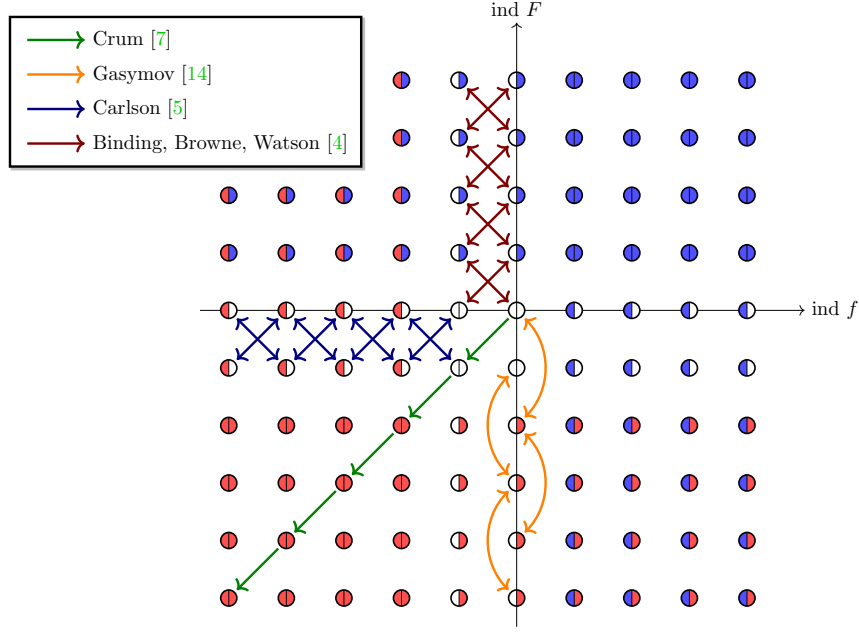
\begin{figure}[ht]
  \resizebox{0.9\textwidth}{!}{\begin{tikzpicture}
    \definecolor{singular}{RGB}{255,80,80}
    \definecolor{constant}{RGB}{255,255,255}
    \definecolor{dependent}{RGB}{80,80,255}
    \definecolor{crum}{RGB}{0,128,0}
    \definecolor{gasymov}{RGB}{255,128,0}
    \definecolor{carlson}{RGB}{0,0,128}
    \definecolor{binding}{RGB}{128,0,0}
    \pgfmathsetmacro\latticesize{4}
    \pgfmathsetmacro\dotradius{0.14}

    \draw[->] (-\latticesize-1.5,0) -- (\latticesize+1,0) node[right] {ind $f$};
    \draw[->] (0,-\latticesize-1.5) -- (0,\latticesize+1) node[above] {ind $F$};

    \foreach \y in {-5,...,\latticesize} {
      \foreach \x in {-5,...,-2} {
        \fill[singular] (\x,\y+\dotradius) arc (90:270:\dotradius);
        \fill[singular] (\y,\x-\dotradius) arc (270:450:\dotradius);}
      \foreach \x in {-1,0} {
        \fill[constant] (\x,\y+\dotradius) arc (90:270:\dotradius);
        \fill[constant] (\y,\x-\dotradius) arc (270:450:\dotradius);}
      \foreach \x in {1,...,\latticesize} {
        \fill[dependent] (\x,\y+\dotradius) arc (90:270:\dotradius);
        \fill[dependent] (\y,\x-\dotradius) arc (270:450:\dotradius);}}

    \foreach \y in {-5,...,\latticesize}
      \foreach \x in {-5,...,\latticesize} {
        \draw[black, thick] (\x,\y) circle (\dotradius);
        \draw[black, thin] (\x,\y-\dotradius) -- (\x,\y+\dotradius);}

    \foreach \x in {-1,...,-\latticesize} {
      \draw[carlson, very thick, <->] (\x-\dotradius,-\dotradius) -- (\x-1+\dotradius,-1+\dotradius);
      \draw[carlson, very thick, <->] (\x-\dotradius,-1+\dotradius) -- (\x-1+\dotradius,-\dotradius);}

    \foreach \x in {0,2}
      \draw[gasymov, very thick, <->] (\dotradius,-\x-2+\dotradius) to [out=45,in=315] (\dotradius,-\x-\dotradius);
    \foreach \x in {1,3}
      \draw[gasymov, very thick, <->] (-\dotradius,-\x-2+\dotradius) to [out=135,in=225] (-\dotradius,-\x-\dotradius);

    \foreach \x in {0,...,\latticesize}
      \draw[crum, very thick, ->] (-\x-\dotradius,-\x-\dotradius) -- (-\x-1+\dotradius,-\x-1+\dotradius);

    \foreach \x in {0,...,3} {
      \draw[binding, very thick, <->] (-\dotradius,\x+\dotradius) -- (-1+\dotradius,\x+1-\dotradius);
      \draw[binding, very thick, <->] (-1+\dotradius,\x+\dotradius) -- (-\dotradius,\x+1-\dotradius);}

    \pgfmathsetmacro\legendx{-\latticesize-4.8}
    \pgfmathsetmacro\legendy{\latticesize-1.5}
    \pgfmathsetmacro\legendwidth{6.1}
    \pgfmathsetmacro\legendheight{2.6}

    \draw[lightgray, ultra thick] (\legendx+0.04, \legendy-0.04) rectangle (\legendx+\legendwidth+0.04, \legendy+\legendheight-0.04);
    \filldraw[fill=white, draw=black, very thick] (\legendx, \legendy) rectangle (\legendx+\legendwidth, \legendy+\legendheight);

    \draw[crum, very thick, ->] (\legendx+0.3,\legendy+2.2) -- (\legendx+1.3,\legendy+2.2) node[black, right] {Crum~\cite{C55}};
    \draw[gasymov, very thick, ->] (\legendx+0.3,\legendy+1.6) -- (\legendx+1.3,\legendy+1.6) node[black, right] {Gasymov~\cite{G65}};
    \draw[carlson, very thick, ->] (\legendx+0.3,\legendy+1) -- (\legendx+1.3,\legendy+1) node[black, right] {Carlson~\cite{C93}};
    \draw[binding, very thick, ->] (\legendx+0.3,\legendy+0.4) -- (\legendx+1.3,\legendy+0.4) node[black, right] {Binding, Browne, Watson~\cite{BBW02b}};
  \end{tikzpicture}}
  \caption{Single and double commutation transformations used in some of the previous works}
  \label{fig:previous}
\end{figure}

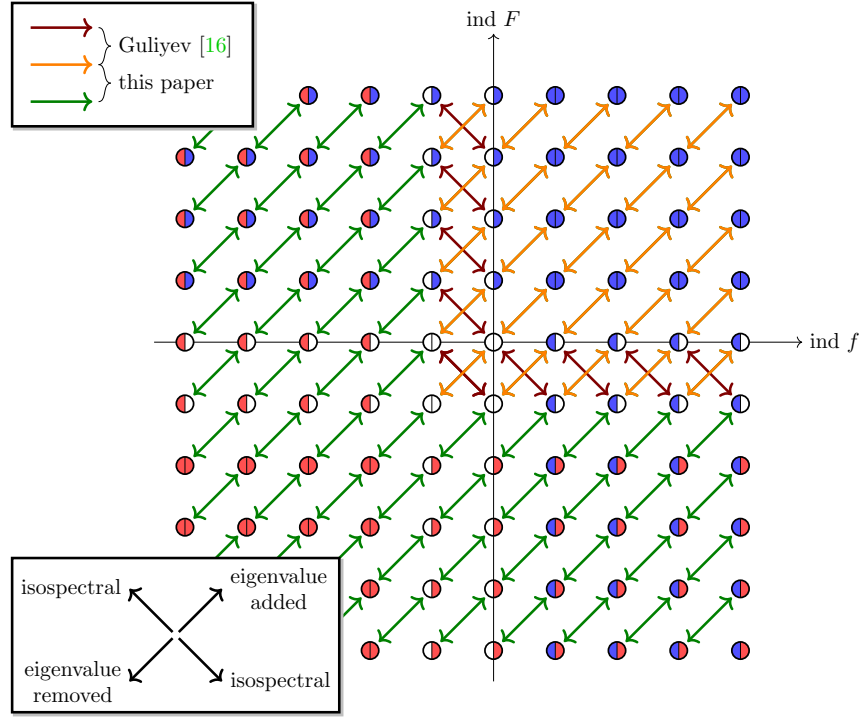
\begin{figure}[ht]
  \resizebox{0.9\textwidth}{!}{\begin{tikzpicture}
    \definecolor{singular}{RGB}{255,80,80}
    \definecolor{constant}{RGB}{255,255,255}
    \definecolor{dependent}{RGB}{80,80,255}
    \definecolor{annali}{RGB}{128,0,0}
    \definecolor{this}{RGB}{0,128,0}
    \definecolor{common}{RGB}{255,128,0}
    \pgfmathsetmacro\latticesize{4}
    \pgfmathsetmacro\dotradius{0.14}

    \draw[->] (-\latticesize-1.5,0) -- (\latticesize+1,0) node[right] {ind $f$};
    \draw[->] (0,-\latticesize-1.5) -- (0,\latticesize+1) node[above] {ind $F$};

    \foreach \y in {-5,...,\latticesize} {
      \foreach \x in {-5,...,-2} {
        \fill[singular] (\x,\y+\dotradius) arc (90:270:\dotradius);
        \fill[singular] (\y,\x-\dotradius) arc (270:450:\dotradius);}
      \foreach \x in {-1,0} {
        \fill[constant] (\x,\y+\dotradius) arc (90:270:\dotradius);
        \fill[constant] (\y,\x-\dotradius) arc (270:450:\dotradius);}
      \foreach \x in {1,...,\latticesize} {
        \fill[dependent] (\x,\y+\dotradius) arc (90:270:\dotradius);
        \fill[dependent] (\y,\x-\dotradius) arc (270:450:\dotradius);}}

    \foreach \y in {-5,...,\latticesize}
      \foreach \x in {-5,...,\latticesize} {
        \draw[black, thick] (\x,\y) circle (\dotradius);
        \draw[black, thin] (\x,\y-\dotradius) -- (\x,\y+\dotradius);}

    \foreach \x in {-5,...,3}
      \foreach \y in {-5,...,3}
        \draw[this, very thick, <->] (\x+\dotradius,\y+\dotradius) -- (\x+1-\dotradius,\y+1-\dotradius);

    \foreach \x in {-1,...,3} {
      \draw[annali, very thick, <->] (\x+\dotradius,-\dotradius) -- (\x+1-\dotradius,-1+\dotradius);
      \draw[annali, very thick, <->] (-\dotradius,\x+\dotradius) -- (-1+\dotradius,\x+1-\dotradius);}

    \foreach \x in {-1,...,3}
      \foreach \y in {-1,...,3}
        \draw[common, very thick, <->] (\x+\dotradius,\y+\dotradius) -- (\x+1-\dotradius,\y+1-\dotradius);

    \pgfmathsetmacro\legendx{-\latticesize-3.8}
    \pgfmathsetmacro\legendy{\latticesize-0.5}
    \pgfmathsetmacro\legendwidth{3.9}
    \pgfmathsetmacro\legendheight{2}

    \draw[lightgray, ultra thick] (\legendx+0.04, \legendy-0.04) rectangle (\legendx+\legendwidth+0.04, \legendy+\legendheight-0.04);
    \filldraw[fill=white, draw=black, very thick] (\legendx, \legendy) rectangle (\legendx+\legendwidth, \legendy+\legendheight);

    \draw[annali, very thick, ->] (\legendx+0.3,\legendy+1.6) -- (\legendx+1.3,\legendy+1.6);
    \draw[common, very thick, ->] (\legendx+0.3,\legendy+1.0) -- (\legendx+1.3,\legendy+1.0);
    \draw[this, very thick, ->] (\legendx+0.3,\legendy+0.4) -- (\legendx+1.3,\legendy+0.4);
    \draw (\legendx+1.4,\legendy+1.6) to [out=0,in=180] (\legendx+1.6,\legendy+1.3);
    \draw (\legendx+1.4,\legendy+1.0) to [out=0,in=180] (\legendx+1.6,\legendy+1.3) node[right] {Guliyev~\cite{G17}};
    \draw (\legendx+1.4,\legendy+1.0) to [out=0,in=180] (\legendx+1.6,\legendy+0.7);
    \draw (\legendx+1.4,\legendy+0.4) to [out=0,in=180] (\legendx+1.6,\legendy+0.7) node[right] {this paper};

    \pgfmathsetmacro\lgndtwox{-\latticesize-3.8}
    \pgfmathsetmacro\lgndtwoy{-\latticesize-2}
    \pgfmathsetmacro\lgndtwowidth{5.3}
    \pgfmathsetmacro\lgndtwoheight{2.5}

    \draw[lightgray, ultra thick] (\lgndtwox+0.04, \lgndtwoy-0.04) rectangle (\lgndtwox+\lgndtwowidth+0.04, \lgndtwoy+\lgndtwoheight-0.04);
    \filldraw[fill=white, draw=black, very thick] (\lgndtwox, \lgndtwoy) rectangle (\lgndtwox+\lgndtwowidth, \lgndtwoy+\lgndtwoheight);

    \draw[very thick, ->] (\lgndtwox+2.6,\lgndtwoy+1.3) -- (\lgndtwox+1.9,\lgndtwoy+2.0) node[left, align=center] {isospectral};
    \draw[very thick, ->] (\lgndtwox+2.7,\lgndtwoy+1.3) -- (\lgndtwox+3.4,\lgndtwoy+2.0) node[right, align=center] {eigenvalue \\ added};
    \draw[very thick, ->] (\lgndtwox+2.7,\lgndtwoy+1.2) -- (\lgndtwox+3.4,\lgndtwoy+0.5) node[right, align=center] {isospectral};
    \draw[very thick, ->] (\lgndtwox+2.6,\lgndtwoy+1.2) -- (\lgndtwox+1.9,\lgndtwoy+0.5) node[left, align=center] {eigenvalue \\ removed};
  \end{tikzpicture}}
  \caption{Transformations considered in~\cite{G17} and in this paper}
  \label{fig:ours}
\end{figure}

We can use the lattice of integer points in the plane to schematically describe boundary value problems and transformations between them, where the abscissa (respectively, the ordinate) of each point equals the index of the boundary condition or singularity at the left (respectively, the right) endpoint. For instance, the transformations considered in some of the papers mentioned above are shown in Figure~\ref{fig:previous}. Single commutations are indicated by diagonal arrows and double commutations by (slightly curved) vertical arrows. Figure~\ref{fig:ours} shows the transformations considered in this paper together with those considered in~\cite{G17}. Transformations (arrows) pointing in the south-west direction remove the smallest eigenvalue, those pointing in the north-east direction add a new eigenvalue below the spectrum, and finally, ones pointing in the south-east or north-west directions (absent from this paper) are isospectral.

The paper is organized as follows. Section~\ref{sec:preliminaries} contains the necessary notation and definitions, as well as some preliminary results. Section~\ref{sec:transformations} is devoted to our transformations. In Subsection~\ref{ss:nevanlinna} we define two transformations on the set of rational Herglotz--Nevanlinna functions and inverse square singularities. In Subsections~\ref{ss:isospectral} and~\ref{ss:inverseisospectral} we define two transformations between boundary value problems, study how these transformations affect the eigenvalues and the norming constants of a problem, and show that they are inverses of each other. We apply these transformations in Section~\ref{sec:applications} to the solution of a number of direct and inverse spectral problems. Namely, we study the asymptotics of the eigenvalues and the norming constants in Subsection~\ref{ss:asymptotics}, prove an oscillation theorem in Subsection~\ref{ss:oscillation}, calculate the so-called \emph{regularized trace} in Subsection~\ref{ss:trace}, obtain necessary and sufficient conditions for two sequences of real numbers to be the eigenvalues and the norming constants of some boundary value problem in Subsection~\ref{ss:byspectraldata}, and consider symmetric problems and a Hochstadt--Lieberman-type result in the final Subsection~\ref{ss:byonespectrum}.

\section{Preliminaries} \label{sec:preliminaries}

We start by recalling some definitions from~\cite{G17} and introducing some new ones. We also present several preliminary results which will be needed in the subsequent sections.

\subsection{Notation} \label{ss:notation}

Every rational Herglotz--Nevanlinna function can be written as
\begin{equation*}
  f(\lambda) = h_0 \lambda + h + \sum_{k=1}^d \frac{\delta_k}{h_k - \lambda}
\end{equation*}
with $h_0 \ge 0$, $h \in \mathbb{R}$, $\delta_k > 0$, and $h_1 < \ldots < h_d$. We assign to each function $f$ of this form its \emph{index}
\begin{equation*}
  \ind f := \begin{cases} 2 d + 1, & h_0 > 0, \\ 2 d, & h_0 = 0 \end{cases}
\end{equation*}
and two polynomials $f_\uparrow$ and $f_\downarrow$ by writing this function as
\begin{equation*}
  f(\lambda) = \frac{f_\uparrow(\lambda)}{f_\downarrow(\lambda)},
\end{equation*}
where
\begin{equation*}
  f_\downarrow(\lambda) := h'_0 \prod_{k=1}^d (h_k - \lambda), \qquad h'_0 := \begin{cases} 1 / h_0, & h_0 > 0, \\ 1, & h_0 = 0. \end{cases}
\end{equation*}
Clearly, the index counts each finite pole of a rational Herglotz--Nevanlinna function twice and its pole at infinity (if any) once. We denote the smallest pole of $f$ and the total number of its poles (if it has any) not exceeding $\lambda$ by $\mathring{\boldsymbol{\uppi}}(f)$ and $\boldsymbol{\Pi}_f(\lambda)$ respectively:
\begin{equation*}
  \mathring{\boldsymbol{\uppi}}(f) := \begin{cases} h_1, & \ind f \ge 2, \\ +\infty, & \ind f \le 1, \end{cases} \qquad\qquad \boldsymbol{\Pi}_f(\lambda) := \sum_{\substack{1 \le k \le d \\ h_k \le \lambda}} 1.
\end{equation*}

In~\cite{G17} we used the infinity symbol to represent the Dirichlet condition. In this paper we denote the Dirichlet condition by $\boldsymbol{\infty}_0$, and introduce an entire sequence of ``infinities'' $\boldsymbol{\infty}_n$ to denote singularities with coefficients $n (n + 1)$ for $n \ge 1$. We set $\ind \boldsymbol{\infty}_n := - n - 1$. We denote by $\mathfrak{B}$ the set of all rational Herglotz--Nevanlinna functions and these infinity symbols $\boldsymbol{\infty}_n$, $n \ge 0$. The above two functions are extended to all $f \in \mathfrak{B}$ by setting $\mathring{\boldsymbol{\uppi}}(\boldsymbol{\infty}_n) = +\infty$ and $\boldsymbol{\Pi}_{\boldsymbol{\infty}_n}(\lambda) \equiv 0$.

In order to be able to deal with boundary conditions and inverse square singularities in a unified way, we introduce the notation
\begin{equation*}
  \ell_f := -1 - \min \{ 0, \ind f \}
\end{equation*}
for $f \in \mathfrak{B}$. Obviously, $\ell_f = n$ if $f = \boldsymbol{\infty}_n$ and $\ell_f = -1$ if $f$ is a rational Herglotz--Nevanlinna function. Thus, in effect, we are just assigning the ``orbital quantum number'' $-1$ to all rational Herglotz--Nevanlinna functions.

We denote by $\mathscr{P}(q, f, F)$ the boundary value problem generated by the equation
\begin{equation} \label{eq:SL}
  -y''(x) + \left( \frac{\ell_f (\ell_f + 1)}{x^2} + \frac{\ell_F (\ell_F + 1)}{(\pi - x)^2} + q(x) \right) y(x) = \lambda y(x)
\end{equation}
and the boundary conditions
\begin{equation} \label{eq:boundary}
  \frac{y'(0)}{y(0)} = -f(\lambda), \qquad \frac{y'(\pi)}{y(\pi)} = F(\lambda)
\end{equation}
at nonsingular endpoints, if any. In other words, at each endpoint we either have a boundary condition or inverse square singularity. We start with the one-dimensional Schr\"{o}dinger equation with the potential $q$, and then adjoin the first boundary condition from~(\ref{eq:boundary}) (respectively, the Dirichlet condition $y(0) = 0$) if $f$ is a rational Herglotz--Nevanlinna function (respectively, $f = \boldsymbol{\infty}_0$), or simply add the term $n (n + 1) / x^2$ to the potential if $f = \boldsymbol{\infty}_n$ with $n \ge 1$, in which case no boundary condition is needed as the potential becomes limit-point at $0$. We do the same for the right endpoint.

We denote by $\mathscr{W}_2^1[0, \pi]$ the Sobolev space of absolutely continuous functions $y$ with $y' \in \mathscr{L}_2(0, \pi)$. The notation
\begin{equation*}
  x_n = y_n + \ell_2 \left( \frac{1}{n^\alpha} \right)
\end{equation*}
means $\sum_{n = 0}^\infty \left| n^\alpha (x_n - y_n) \right|^2 < \infty$, and $\mathbf{1}_A$ denotes the indicator function of a set $A$:
\begin{equation*}
  \mathbf{1}_A(x) := \begin{cases} 1, & x \in A, \\ 0, & x \notin A. \end{cases}
\end{equation*}

\subsection{Regular solutions} \label{ss:solutions}

We define the \emph{left regular solution} $\varphi$ of the problem $\mathscr{P}(q, f, F)$ as the solution of the equation~(\ref{eq:SL}) that satisfies the initial conditions
\begin{equation*}
  \varphi(0, \lambda) = f_\downarrow(\lambda), \qquad \varphi'(0, \lambda) = -f_\uparrow(\lambda)
\end{equation*}
when $\ind f \ge -1$ and asymptotically behaves as
\begin{equation*}
  \varphi(x, \lambda) \sim \frac{x^{\ell_f + 1}}{\left( 2 \ell_f + 1 \right)!!}, \quad x \to 0
\end{equation*}
near the left endpoint when $\ind f \le -1$, where $\left( 2 \ell_f + 1 \right)!! := \prod_{k=0}^{\ell_f} (2 k + 1)$. For the Dirichlet boundary condition ($\ind f = -1$) these two conditions coincide. Similarly, the \emph{right regular solution} $\psi$ is defined as the one satisfying the initial conditions
\begin{equation*}
  \psi(\pi, \lambda) = F_\downarrow(\lambda), \qquad \psi'(\pi, \lambda) = F_\uparrow(\lambda)
\end{equation*}
or the asymptotics
\begin{equation*}
  \psi(x, \lambda) \sim \frac{(\pi - x)^{\ell_F + 1}}{\left( 2 \ell_F + 1 \right)!!}, \quad x \to \pi.
\end{equation*}

For each fixed $x \in (0, \pi)$ (and also for nonsingular endpoints) the solutions $\varphi(x, \lambda)$ and $\psi(x, \lambda)$ together with their first derivatives with respect to $x$ are entire functions of $\lambda$. The boundary value problem $\mathscr{P}(q, f, F)$ has a discrete set of eigenvalues, which are real and simple, and for each eigenvalue $\lambda_n$ there exists a unique number $\beta_n \ne 0$ such that
\begin{equation*}
  \psi(x, \lambda_n) = \beta_n \varphi(x, \lambda_n).
\end{equation*}
These eigenvalues coincide with the zeros of the \emph{characteristic function}
\begin{equation*}
  \chi(\lambda) := \varphi(x, \lambda) \psi'(x, \lambda) - \varphi'(x, \lambda) \psi(x, \lambda),
\end{equation*}
which is independent of $x \in (0, \pi)$. The asymptotics of the regular solutions and their first derivatives (see \cite[formulas (2.24), (2.25)]{KST10} and the proof of \cite[Lemma 2.2]{G17} for the cases of singular and nonsingular endpoints respectively) show that the characteristic function is entire of order $1/2$. Therefore, Hadamard's theorem \cite[Theorem 9.10.8]{S15} yields the infinite product representation
\begin{equation} \label{eq:product}
  \chi(\lambda) = -\pi^{\mathbf{1}_{\mathbb{Z}}(L)} \prod_{n = \lfloor L \rfloor + 1}^{n = -1} \frac{1}{(n - L)^2} \prod_{n = 0}^{n = \lfloor L \rfloor} (\lambda_n - \lambda) \prod_{n = \max \{ \lfloor L \rfloor + 1, 0 \}}^{+\infty} \frac{\lambda_n - \lambda}{(n - L)^2}
\end{equation}
with $L := (\ind f + \ind F) / 2$; see, e.g., \cite[Lemma A.1]{G18a} for details.

If the left endpoint is singular or the boundary condition at the left endpoint is Dirichlet (i.e., $\ell_f \ge 0$) then the left regular solution satisfies the asymptotics
\begin{equation*}
  \frac{\varphi'(x, \lambda)}{\varphi(x, \lambda)} = \frac{\ell_f + 1}{x} + o(1), \quad x \to 0.
\end{equation*}
We will also need in Subsection \ref{ss:trace} a more refined version of this asymptotics when $q \in \mathscr{W}_2^1[0, \pi]$. In this case we have
\begin{equation} \label{eq:fine_asymptotics}
  \frac{\varphi'(x, \lambda)}{\varphi(x, \lambda)} = \frac{\ell_f + 1}{x} + \frac{x}{2 \ell_f + 3} \left( q(0) + \frac{\ell_F (\ell_F + 1)}{\pi^2} - \lambda \right) + o(x), \quad x \to 0.
\end{equation}
Both of these asymptotics can be obtained from the proof of \cite[Theorem 4.2]{C93}.

The following nonoscillation result can be proved in the usual way (see, e.g., the proofs of \cite[Lemmas 3.3 and 3.4]{C93} and \cite[Lemma 2.5]{G17} for details).

\begin{lemma} \label{lem:no_zero}
  If $\lambda \le \lambda_0$ then the left regular solution $\varphi(x, \lambda)$ and the right regular solution $\psi(x, \lambda)$ are strictly positive in $(0,\pi)$.
\end{lemma}

Arguing as in the proofs of \cite[Theorems 4.2 and 4.3]{C93} one can also prove that if $y$ is a solution of~(\ref{eq:SL}) having no zeros in $(0,\pi)$, then $\left( y' / y \right)^2 - m_f^2 / x^2 - m_F^2 / (\pi - x)^2 \in \mathscr{L}_2(0, \pi)$, where $m_f = \ell_f + 1$ or $m_f = -\ell_f$, depending on whether $y$ is bounded near the left endpoint or not, and similarly for the right endpoint. Together with the obvious identity
\begin{equation} \label{eq:Riccati}
  \left( \frac{y'(x, \lambda)}{y(x, \lambda)} \right)' + \left( \frac{y'(x, \lambda_0)}{y(x, \lambda_0)} \right)^2 = q(x) + \frac{\ell_f (\ell_f + 1)}{x^2} + \frac{\ell_F (\ell_F + 1)}{(\pi - x)^2} - \lambda
\end{equation}
this gives $\left( y' / y \right)' + m_f / x^2 + m_F / (\pi - x)^2 \in \mathscr{L}_2(0, \pi)$. It is immediate from Lemma~\ref{lem:no_zero} that the condition of not having zeros in $(0,\pi)$ is certainly fulfilled for $\varphi$, and more generally, for linear combinations of $\varphi$ and $\psi$ with strictly positive coefficients. For the sake of later reference we now formulate our result in these two special cases as a lemma.

\begin{lemma} \label{lem:L_2}
  The eigenfunction $\varphi(x, \lambda_0)$ corresponding to the smallest eigenvalue satisfies the relation
\begin{equation*}
  \left( \frac{\varphi'(x, \lambda_0)}{\varphi(x, \lambda_0)} \right)' + \frac{\ell_f + 1}{x^2} + \frac{\ell_F + 1}{(\pi - x)^2} \in \mathscr{L}_2(0, \pi).
\end{equation*}
Similarly, if $\lambda < \lambda_0$ and $y(x, \lambda)$ is a solution of the equation~(\ref{eq:SL}) that can be expressed as a linear combination with strictly positive coefficients of the regular solutions $\varphi$ and $\psi$, then
\begin{equation*}
  \left( \frac{y'(x, \lambda)}{y(x, \lambda)} \right)' - \frac{\min \{ \ell_f, 0 \}}{x^2} - \frac{\min \{ \ell_F, 0 \}}{(\pi - x)^2} \in \mathscr{L}_2(0, \pi).
\end{equation*}
\end{lemma}

The use of the minimum function in the second case stems from the following observation: for values of $\lambda$ not coinciding with one of the eigenvalues, a (left or right) regular solution is unbounded near the opposite endpoint if and only if that endpoint is singular.

\subsection{Hilbert space formulation and norming constants} \label{ss:hilbert}

The differential expression on the left-hand side of~(\ref{eq:SL}) generates a self-adjoint operator in $\mathscr{L}_2(0, \pi)$ whenever there are no boundary conditions dependent on the eigenvalue parameter. The domain of this operator consists of functions for which the differential expression in~(\ref{eq:SL}) is well defined and which satisfy the corresponding boundary conditions (if any). This is no longer true when the eigenvalue parameter is present in one or both of the boundary conditions. One now needs to consider a direct sum of $\mathscr{L}_2(0, \pi)$ and a suitable number of copies of $\mathbb{C}$: $\mathfrak{H} = \mathscr{L}_2(0,\pi) \oplus \mathbb{C}^N$. The number of these extra copies equals the total number of poles of the functions $f$ and $F$, including those at infinity:
\begin{equation*}
  N := \left\lceil \frac{\ind f}{2} \right\rceil_{+} + \left\lceil \frac{\ind F}{2} \right\rceil_{+},
\end{equation*}
where we denoted by $\lceil \cdot \rceil_{+}$ the smallest nonnegative integer not smaller than the argument.

We now describe the operator when the left endpoint is singular ($\ind f \le -2$) and there is dependence on the eigenvalue parameter at the right endpoint ($\ind F \ge 1$ is odd), i.e.,
\begin{equation*}
  F(\lambda) = H_0 \lambda + H + \sum_{k=1}^D \frac{\Delta_k}{H_k - \lambda}
\end{equation*}
with $H_0 > 0$, $H \in \mathbb{R}$, $\Delta_k > 0$, and $H_1 < \ldots < H_D$. In this case, $\mathfrak{H} = \mathscr{L}_2(0,\pi) \oplus \mathbb{C}^{D+1}$ with inner product given by
\begin{equation*}
  \langle Y, Z \rangle := \int_0^{\pi} y(x) \overline{z(x)} \,\du x + \sum_{k=1}^{D} \frac{\eta_k \overline{\zeta_k}}{\Delta_k} + \frac{\eta_{D+1} \overline{\zeta_{D+1}}}{H_0}
\end{equation*}
for $Y = \begin{pmatrix} y & \eta_1 & \dots & \eta_{D+1} \end{pmatrix}^T \in \mathfrak{H}$ and $Z = \begin{pmatrix} z & \zeta_1 & \dots & \zeta_{D+1} \end{pmatrix}^T \in \mathfrak{H}$, where the superscript ${}^T$ denotes the transpose. The self-adjoint operator in this space is defined by
\begin{equation*}
  A(Y) := \begin{pmatrix} -y''(x) + \left( \frac{\ell_f (\ell_f + 1)}{x^2} + q(x) \right) y(x) \\ H_1 \eta_1 - \Delta_1 y(\pi) \\ \vdots \\ H_{D} \eta_{D} - \Delta_{D} y(\pi) \\ y'(\pi) - H y(\pi) - \sum_{k=1}^{D} \eta_k \end{pmatrix}
\end{equation*}
on the domain
\begin{multline*}
  \mathfrak{D}(A) := \left\{ Y \in \mathfrak{H} \bigm| y, y' \in \mathscr{AC}(0,\pi],\ -y'' + \left( \ell_f (\ell_f + 1) / x^2 + q \right) y \in \mathscr{L}_2(0,\pi), \right. \\
  \left. \vphantom{\bigm|} \eta_{D+1} = H_0 y(\pi) \right\}.
\end{multline*}
The other cases are similar; see, e.g., \cite[Section 2.2]{G17} for the general regular case.

The operator $A$ thus defined is self-adjoint in $\mathfrak{H}$ and its spectrum is purely discrete and coincides with the set of eigenvalues of the boundary value problem $\mathscr{P}(q, f, F)$. We denote by $\Phi_n$ the eigenvector of this operator whose first component coincides with $\varphi(x, \lambda_n)$. For instance, in the case considered in the preceding paragraph $\Phi_n$ has the form
\begin{equation*}
  \Phi_n := \begin{pmatrix} \varphi(x, \lambda_n) & \frac{\Delta_1}{H_1 - \lambda_n} \varphi(\pi, \lambda_n) & \dots & \frac{\Delta_{D}}{H_{D} - \lambda_n} \varphi(\pi, \lambda_n) & H_0 \varphi(\pi, \lambda_n) \end{pmatrix}^T.
\end{equation*}
These eigenvectors form an orthogonal basis for $\mathfrak{H}$. As for the eigenfunctions $\varphi(x, \lambda_n)$ of the boundary value problem $\mathscr{P}(q, f, F)$, they may form a basis for $\mathscr{L}_2(0,\pi)$ only after removing $N$ of them and that also depends on which ones exactly are removed~\cite{G19b}.

We define the \emph{norming constants} as
\begin{equation*}
  \gamma_n := \| \Phi_n \|_{\mathfrak{H}}^2.
\end{equation*}
For a problem without singular endpoints, the three sequences $\{ \lambda_n \}_{n \ge 0}$, $\{ \beta_n \}_{n \ge 0}$, and $\{ \gamma_n \}_{n \ge 0}$ satisfy the identity \cite[Lemma~2.1]{G17}
\begin{equation} \label{eq:chi_beta_gamma}
  \chi'(\lambda_n) = \beta_n \gamma_n.
\end{equation}
That this identity also holds for problems with singular endpoints can be verified in a straightforward manner by using the above asymptotics of the regular solutions and the fact that any solution of the equation~(\ref{eq:SL}) behaves as $O \left( x^{-\ell_f} \right)$ near (say) the left endpoint.

\section{Transformations} \label{sec:transformations}

We now turn to our transformations. First we define transformations on the set of rational Herglotz--Nevanlinna functions and inverse square singularities. Then we introduce the main players of this paper---transformations between boundary value problems of the form $\mathscr{P}(q, f, F)$.

\subsection{Transformation of boundary conditions and singularities} \label{ss:nevanlinna}

In \cite{G17}, we used a single transformation to both increase and decrease the index of a boundary condition. For the purposes of this paper, however, it is more convenient to define two separate transformations, one in either direction. The first of these transformations acts in the direction from eigenparameter dependent boundary conditions towards inverse square singularities, while the second in the opposite direction (see Figure~\ref{fig:index}).

On the set $\widehat{\Sigma} := \left\{ (\mu, f) \in \mathbb{R} \times \mathfrak{B} \colon \mu < \mathring{\boldsymbol{\uppi}}(f) \right\}$, we define the transformation
\begin{equation*}
  \widehat{\boldsymbol{\Theta}} \colon \widehat{\Sigma} \to \mathfrak{B},\ (\mu, f) \mapsto \widehat{f}
\end{equation*}
by
\begin{equation*}
  \widehat{f} := \begin{cases} \lambda \mapsto \frac{\mu - \lambda}{f(\lambda) - f(\mu)} - f(\mu), & \ind f \ge 1, \\ \boldsymbol{\infty}_{-\ind f}, & \ind f \le 0. \end{cases}
\end{equation*}
Here $\widehat{f}(\mu) = - \left( f'(\mu) \right)^{-1} - f(\mu)$, by definition, in the first case. Similarly, on the set $\widetilde{\Sigma} := \left\{ (\mu, \tau, f) \in \mathbb{R}^2 \times \mathfrak{B} \colon \mu < \mathring{\boldsymbol{\uppi}}(f),\ \tau > f(\mu) \text{ if } \ind f \ge 0 \right\}$, we define the transformation
\begin{equation*}
  \widetilde{\boldsymbol{\Theta}} \colon \widetilde{\Sigma} \to \mathfrak{B},\ (\mu, \tau, f) \mapsto \widetilde{f}
\end{equation*}
by
\begin{equation*}
  \widetilde{f} := \begin{cases} \lambda \mapsto \frac{\mu - \lambda}{f(\lambda) - \tau} - \tau, & \ind f \ge 0, \\ \lambda \mapsto -\tau, & \ind f = -1, \\ \boldsymbol{\infty}_{-\ind f - 2}, & \ind f \le -2. \end{cases}
\end{equation*}
These two transformations are inverses of each other in the sense that
\begin{equation*}
  \widetilde{\boldsymbol{\Theta}}(\mu, -f(\mu), \widehat{\boldsymbol{\Theta}}(\mu, f)) = f, \qquad \widehat{\boldsymbol{\Theta}}(\mu, \widetilde{\boldsymbol{\Theta}}(\mu, \tau, f)) = f
\end{equation*}
for any $(\mu, f) \in \widehat{\Sigma}$ and $(\mu, \tau, f) \in \widetilde{\Sigma}$ respectively. The following result, the proof of which can be found in \cite[Subsection 3.1]{G17}, implies in particular that the transformations $\widehat{\boldsymbol{\Theta}}$ and $\widetilde{\boldsymbol{\Theta}}$ are well defined.

\begin{lemma} \label{lem:f_hat}
  We have $\widehat{f} := \widehat{\boldsymbol{\Theta}}(\mu, f) \in \mathfrak{B}$ with $\ind \widehat{f} = \ind f - 1$ and $\widetilde{f} := \widetilde{\boldsymbol{\Theta}}(\mu, \tau, f) \in \mathfrak{B}$ with $\ind \widetilde{f} = \ind f + 1$. The poles of $f$ and $\widehat{f}$ (respectively, $f$ and $\widetilde{f}$) interlace each other with $\mathring{\boldsymbol{\uppi}}(f) < \mathring{\boldsymbol{\uppi}}(\widehat{f})$ (respectively, $\mathring{\boldsymbol{\uppi}}(f) > \mathring{\boldsymbol{\uppi}}(\widetilde{f})$) whenever both of them have poles, i.e., when $\ind f \ge 3$ (respectively, when $\ind f \ge 2$). Moreover,
\begin{equation*}
  \widehat{f}_\uparrow(\lambda) = \frac{- f(\mu) f_\uparrow(\lambda) - \left( \lambda - \mu - f^2(\mu) \right) f_\downarrow(\lambda)}{\lambda - \mu}, \qquad \widehat{f}_\downarrow(\lambda) = \frac{f_\uparrow(\lambda) - f(\mu) f_\downarrow(\lambda)}{\lambda - \mu}
\end{equation*}
and
\begin{equation*}
  \widetilde{f}_\uparrow(\lambda) = \tau f_\uparrow(\lambda) + \left( \lambda - \mu - \tau^2 \right) f_\downarrow(\lambda), \qquad \widetilde{f}_\downarrow(\lambda) = -f_\uparrow(\lambda) + \tau f_\downarrow(\lambda)
\end{equation*}
for any rational Herglotz--Nevanlinna function $f$.
\end{lemma}

\subsection{Transformation of problems which removes the smallest eigenvalue} \label{ss:isospectral}

Now we introduce our first transformation between boundary value problems of the form~$\mathscr{P}(q, f, F)$ (depicted as an arrow pointing in the south-west direction in Figure~\ref{fig:ours}). This transformation decreases the indices of both $f$ and $F$ by one. Therefore, by applying it a sufficient number of times, one eventually arrives at a problem with only inverse square singularities at both endpoints (i.e., a dot located in the third quadrant in Figure~\ref{fig:ours}).

We define the transformation
\begin{equation*}
  \widehat{\mathbf{T}} \colon \mathscr{L}_2(0, \pi) \times \mathfrak{B}^2 \to \mathscr{L}_2(0, \pi) \times \mathfrak{B}^2,\ (q, f, F) \mapsto (\widehat{q}, \widehat{f}, \widehat{F})
\end{equation*}
by
\begin{equation} \label{eq:q_f_F_hat}
\begin{gathered}
  \widehat{f} := \widehat{\boldsymbol{\Theta}} (\lambda_0, f), \qquad \widehat{F} := \widehat{\boldsymbol{\Theta}} (\lambda_0, F), \\
  \widehat{q}(x) := q(x) - 2 \left( \frac{\varphi'(x, \lambda_0)}{\varphi(x, \lambda_0)} \right)' - \frac{2 (\ell_f + 1)}{x^2} - \frac{2 (\ell_F + 1)}{(\pi - x)^2}.
\end{gathered}
\end{equation}
It is immediate from Lemma~\ref{lem:L_2} that this transformation is well defined. The following result shows how exactly the eigenvalues and the norming constants of a boundary value problem are affected under the transformation $\widehat{\mathbf{T}}$.

\begin{theorem} \label{thm:transformation}
  If $(\widehat{q}, \widehat{f}, \widehat{F}) = \widehat{\mathbf{T}} (q, f, F)$ then the eigenvalues $\widehat{\lambda}_n$ and the norming constants $\widehat{\gamma}_n$ of the transformed problem $\mathscr{P}(\widehat{q}, \widehat{f}, \widehat{F})$ are related to the eigenvalues $\lambda_n$ and the norming constants $\gamma_n$ of the original problem $\mathscr{P}(q, f, F)$ by
\begin{equation*}
  \widehat{\lambda}_n = \lambda_{n+1}, \qquad \widehat{\gamma}_n = \frac{\gamma_{n+1}}{\lambda_{n+1} - \lambda_0}, \qquad n \ge 0.
\end{equation*}
\end{theorem}
\begin{proof}
A standard calculation shows that
\begin{equation*}
  \widehat{\varphi}(x, \lambda) := \begin{cases} \frac{1}{\lambda_0 - \lambda} \left( \varphi'(x, \lambda) - \frac{\varphi'(x, \lambda_0)}{\varphi(x, \lambda_0)} \varphi(x, \lambda) \right), & \lambda \ne \lambda_0, \\ - \varphi(x, \lambda_0) \left. \frac{\partial}{\partial\lambda} \left( \frac{\varphi'(x, \lambda)}{\varphi(x, \lambda)} \right) \right|_{\lambda = \lambda_0}, & \lambda = \lambda_0 \end{cases}
\end{equation*}
is the left regular solution of the problem $\mathscr{P}(\widehat{q}, \widehat{f}, \widehat{F})$. It is also straightforward to check that a number $\lambda \ne \lambda_0$ is an eigenvalue of $\mathscr{P}(\widehat{q}, \widehat{f}, \widehat{F})$ if and only if it is an eigenvalue of $\mathscr{P}(q, f, F)$. Therefore one only needs to verify that $\lambda_0$ is not an eigenvalue of $\mathscr{P}(\widehat{q}, \widehat{f}, \widehat{F})$ and this can easily be done by considering the general solution of the equation
\begin{equation} \label{eq:SL_hat}
  -y''(x) + \left( \frac{\ell_{\widehat{f}} (\ell_{\widehat{f}} + 1)}{x^2} + \frac{\ell_{\widehat{F}} (\ell_{\widehat{F}} + 1)}{(\pi - x)^2} + \widehat{q}(x) \right) y(x) = \lambda_0 y(x)
\end{equation}
and obtaining a contradiction (see the proof of \cite[Theorem 3.2]{G19a} for details).

Let us now turn to norming constants. A simple symmetry argument shows that a formula similar to the above, but with a minus sign, also holds for the right regular solutions of the problems $\mathscr{P}(q, f, F)$ and $\mathscr{P}(\widehat{q}, \widehat{f}, \widehat{F})$. Hence $\widehat{\psi}(x, \widehat{\lambda}_n) = - \beta_{n+1} \widehat{\varphi}(x, \widehat{\lambda}_n)$. On the other hand, the infinite product representation (\ref{eq:product}) implies that the characteristic functions of these two problems are related by the formula
\begin{equation} \label{eq:chi_hat}
  \chi(\lambda) = (\lambda_0 - \lambda) \widehat{\chi}(\lambda).
\end{equation}
In particular, $\chi'(\widehat{\lambda}_n) = (\lambda_0 - \widehat{\lambda}_n) \widehat{\chi}'(\widehat{\lambda}_n)$. Therefore, by (\ref{eq:chi_beta_gamma}), the norming constant $\widehat{\gamma}_n$ of $\mathscr{P}(\widehat{q}, \widehat{f}, \widehat{F})$ corresponding to $\widehat{\lambda}_n$ equals
\begin{equation*}
  \frac{\widehat{\chi}'(\widehat{\lambda}_n)}{-\beta_{n+1}} = \frac{\gamma_{n+1}}{\lambda_{n+1} - \lambda_0}.
\end{equation*}
\end{proof}

We will invert the action of the transformation $\widehat{\mathbf{T}}$ in the next subsection. However, as the above theorem shows, the information about the smallest eigenvalue $\lambda_0$ and the corresponding norming constant $\gamma_0$ is lost under this transformation. We will soon see that they can be given arbitrarily, as long as $\lambda_0$ is strictly smaller than the smallest eigenvalue of the problem $\mathscr{P}(\widehat{s}, \widehat{f}, \widehat{F})$ and $\gamma_0$ is positive. To see how these two numbers can help us, we observe from (\ref{eq:q_f_F_hat}) that one needs to know $\varphi(x, \lambda_0)$ to invert the action of $\widehat{\mathbf{T}}$. Since the function $1 / \varphi(x, \lambda_0)$ is also a solution of the equation (\ref{eq:SL_hat}), its Wronskian with $\widehat{\varphi}(x, \lambda_0)$ can be found by calculating its value at $0$ or its limit as $x \to 0$:
\begin{equation*}
  \mathcal{W} \left( \frac{1}{\varphi}, \widehat{\varphi} \right) = 1.
\end{equation*}
Similarly, we have
\begin{equation*}
  \mathcal{W} \left( \frac{1}{\varphi}, \widehat{\psi} \right) = \beta_0 \mathcal{W} \left( \frac{1}{\psi}, \widehat{\psi} \right) = \beta_0.
\end{equation*}
Therefore
\begin{equation*}
  \frac{1}{\varphi} = \frac{1}{\mathcal{W} \left( \widehat{\psi}, \widehat{\varphi} \right)} \left( \widehat{\psi} - \beta_0 \widehat{\varphi} \right).
\end{equation*}
It is this coefficient $\beta_0$ here that can be expressed in terms of $\lambda_0$ and $\gamma_0$. Indeed, by (\ref{eq:chi_beta_gamma}) and (\ref{eq:chi_hat}), we have
\begin{equation*}
  \beta_0 = \frac{\chi'(\lambda_0)}{\gamma_0} = - \frac{\widehat{\chi}(\lambda_0)}{\gamma_0}.
\end{equation*}

\subsection{Transformation of problems which adds a new eigenvalue} \label{ss:inverseisospectral}

We now turn to our second transformation between problems of the form~$\mathscr{P}(q, f, F)$ (depicted as an arrow pointing in the north-east direction in Figure~\ref{fig:ours}). This transformation inverts the action of the transformation $\widehat{\mathbf{T}}$ defined in the previous subsection (see Theorem~\ref{thm:inverse} below) and in particular increases the indices of both $f$ and $F$ by one. Therefore, by applying it a sufficient number of times, one eventually arrives at a problem without any singularities and with eigenparameter dependent (or independent) boundary conditions at both endpoints (i.e., a dot located in the first quadrant in Figure~\ref{fig:ours}).

Throughout this subsection $\mathring{\boldsymbol{\uplambda}}(q, f, F)$ and $\mathring{\boldsymbol{\upgamma}}(q, f, F)$ denote the smallest eigenvalue and the corresponding norming constant, respectively, of a problem $\mathscr{P}(q, f, F)$. We define the transformation
\begin{equation*}
  \widetilde{\mathbf{T}} \colon \widetilde{\mathcal{S}} \to \mathscr{L}_2(0, \pi) \times \mathfrak{B}^2,\ (\mu, \nu, q, f, F) \mapsto (\widetilde{q}, \widetilde{f}, \widetilde{F})
\end{equation*}
on the set
\begin{equation*}
  \widetilde{\mathcal{S}} := \left\{ (\mu, \nu, q, f, F) \in \mathbb{R} \times (0, +\infty) \times \mathscr{L}_2(0, \pi) \times \mathfrak{B}^2 \colon \mu < \mathring{\boldsymbol{\uplambda}}(q, f, F) \right\}
\end{equation*}
by
\begin{equation} \label{eq:q_f_F_tilde}
\begin{gathered}
  \widetilde{f} := \widetilde{\boldsymbol{\Theta}} \left( \mu, -\frac{u'(0)}{u(0)}, f \right), \qquad \widetilde{F} := \widetilde{\boldsymbol{\Theta}} \left( \mu, \frac{u'(\pi)}{u(\pi)}, F \right), \\
  \widetilde{q}(x) := q(x) - 2 \left( \frac{u'(x)}{u(x)} \right)' + \frac{2 (\ell_{\widetilde{f}} + 1)}{x^2} + \frac{2 (\ell_{\widetilde{F}} + 1)}{(\pi - x)^2},
\end{gathered}
\end{equation}
where
\begin{equation*}
  u(x) := \psi(x, \mu) + \frac{\chi(\mu)}{\nu} \varphi(x, \mu).
\end{equation*}
Strictly speaking, the values of the solution $u$ and its first derivative are not defined at singular endpoints, but the transformation $\widetilde{\boldsymbol{\Theta}}$ does not depend on its second argument in this case either. That the transformation $\widetilde{\mathbf{T}}$ is well defined follows from Lemma~\ref{lem:L_2} together with the obvious identity $\min \{ \ell_f, 0 \} = \ell_{\widetilde{f}} + 1$ for any $f \in \mathfrak{B}$.

The analogue of Theorem~\ref{thm:transformation} for the transformation $\widetilde{\mathbf{T}}$ now reads as follows.

\begin{theorem} \label{thm:inverse_transformation}
  If $(\widetilde{q}, \widetilde{f}, \widetilde{F}) = \widetilde{\mathbf{T}}(\mu, \nu, q, f, F)$ then the eigenvalues $\widetilde{\lambda}_n$ and the norming constants $\widetilde{\gamma}_n$ of the transformed problem $\mathscr{P}(\widetilde{q}, \widetilde{f}, \widetilde{F})$ are related to the eigenvalues $\lambda_n$ and the norming constants $\gamma_n$ of the original problem $\mathscr{P}(q, f, F)$ by
\begin{equation*}
  \widehat{\lambda}_0 = \mu, \qquad \widehat{\gamma}_0 = \nu, \qquad\quad \widehat{\lambda}_n = \lambda_{n-1}, \qquad \widehat{\gamma}_n = \gamma_{n-1} (\lambda_{n-1} - \mu), \qquad n \ge 1.
\end{equation*}
\end{theorem}

This theorem is a consequence of the following result, together with Theorem~\ref{thm:transformation} and the discussion thereafter.

\begin{theorem} \label{thm:inverse}
The transformations $\widehat{\mathbf{T}}$ and $\widetilde{\mathbf{T}}$ are inverses of each other in the sense that if $(\widehat{q}, \widehat{f}, \widehat{F}) = \widehat{\mathbf{T}}(q, f, F)$ then $\widetilde{\mathbf{T}} \left( \mathring{\boldsymbol{\uplambda}}(q, f, F), \mathring{\boldsymbol{\upgamma}}(q, f, F), \widehat{q}, \widehat{f}, \widehat{F} \right) = (q, f, F)$ for all $(q, f, F) \in \mathscr{L}_2(0, \pi) \times \mathfrak{B}^2$, and conversely $\widehat{\mathbf{T}} \widetilde{\mathbf{T}}(\mu, \nu, q, f, F) = (q, f, F)$ for all $(\mu, \nu, q, f, F) \in \widetilde{\mathcal{S}}$.
\end{theorem}

The proofs of both theorems are similar to those in \cite[Section 3.3]{G17} and are thus omitted.

\section{Some applications} \label{sec:applications}

As can be seen from Figure~\ref{fig:ours}, by applying one of the transformations defined in the preceding section to any boundary value problem of the form $\mathscr{P}(q, f, F)$ a sufficient number of times, one eventually arrives at either a problem with inverse square singularities at both endpoints or a problem with eigenparameter dependent boundary conditions at both endpoints (i.e., without any singularities). Therefore, these transformations allow one to transfer almost any spectral result to $\mathscr{P}(q, f, F)$ either from problems with eigenparameter dependent boundary conditions or from those with inverse square singularities.

In this section we illustrate these possibilities by generalizing a number of direct and inverse spectral results obtained in~\cite{G17} for problems with eigenparameter dependent boundary conditions to problems with inverse square singularities at one or both endpoints. An exception is made in the second half of the final subsection, where we go the other way around and extend an inverse spectral result from problems with inverse square singularities to problems of the form $\mathscr{P}(q, f, F)$.

Our basic strategy in the first three subsections can be described as follows. Let $\mathscr{P}(q, f, F)$ be an arbitrary problem. We may assume without loss of generality that $K := \max \{ \ell_f, \ell_F \} \ge 1$, since otherwise both of the endpoints would be nonsingular and such problems were already considered in~\cite{G17}. Let $\lambda_{-K} < \ldots < \lambda_{-1}$ be arbitrary numbers such that $\lambda_{-1} < \lambda_0$ and $\gamma_{-K}$, $\ldots$, $\gamma_{-1}$ be arbitrary positive numbers. We set $(q^{(0)}, f^{(0)}, F^{(0)}) := (q, f, F)$ and define the chain of problems $\mathscr{P}(q^{(k)}, f^{(k)}, F^{(k)})$ inductively by
\begin{equation} \label{eq:P_k}
  (q^{(k)}, f^{(k)}, F^{(k)}) := \widehat{\mathbf{T}} (\lambda_{-k}, \gamma_{-k}, q^{(k-1)}, f^{(k-1)}, F^{(k-1)}), \qquad k = 1, 2, \ldots, K.
\end{equation}
Then the last problem $\mathscr{P}(q^{(K)}, f^{(K)}, F^{(K)})$ has no singular endpoints and is thus of the form considered in~\cite{G17}. Now it only remains to observe that the original problem $\mathscr{P}(q, f, F)$ can be reconstructed as $(q, f, F) := \widehat{\mathbf{T}}^K (q^{(K)}, f^{(K)}, F^{(K)})$. In Subsection~\ref{ss:byspectraldata} we consider similar chains consisting not of problems but their supposed eigenvalues and norming constants.

As already mentioned in the introduction, in order to formulate various spectral results in a unified manner, we associated (apart from the index) two other real numbers to every rational Herglotz--Nevanlinna function. We now extend their definition to arbitrary $f \in \mathfrak{B}$ by setting
\begin{equation} \label{eq:omega_Omega}
  \omega_1 = \begin{cases} \frac{1}{h_0}, & \ind f \ge 0 \text{ is odd}, \\ - h, & \ind f \ge 0 \text{ is even}, \\ - \frac{\ell_f (\ell_f + 1)}{2 \pi}, & \ind f \le -1, \end{cases} \quad \omega_2 = \begin{cases} \frac{h}{h_0} - \sum_{k=1}^d h_k, & \ind f \ge 0 \text{ is odd}, \\ - \sum_{k=1}^d h_k, & \ind f \ge 0 \text{ is even}, \\ - \frac{\ell_f^2 (\ell_f + 1)^2}{8 \pi^2}, & \ind f \le -1. \end{cases}
\end{equation}
We define $\Omega_1$ and $\Omega_2$ similarly for the right endpoint.

\begin{remark} \label{rem:omega}
  In~\cite{G17} we associated a monic polynomial $\boldsymbol{\upomega}_f$ to every rational Herglotz--Nevanlinna function $f$ and defined $\omega_1$ and $\omega_2$ as respectively the second and third coefficients of this polynomial. Although we used only these two coefficients of the polynomial in that paper, it subsequently turned out that the whole polynomial $\boldsymbol{\upomega}_f$ is useful in some spectral problems~\cite{G19c}. It would be interesting to generalize this polynomial to our current setting.
\end{remark}

\subsection{Asymptotics of eigenvalues and norming constants} \label{ss:asymptotics}

We have shown in \cite[Subsection 4.1]{G17} that the eigenvalues and the norming constants of a boundary value problem $\mathscr{P}(q, f, F)$ with no singular endpoints (i.e., $\min \{ \ind f, \ind F \} \ge -1$) obey the asymptotics
\begin{equation} \label{eq:lambda}
  \sqrt{\lambda_n} = n - \frac{\ind f + \ind F}{2} + \frac{1}{\pi n} \left( \frac{1}{2} \int_0^\pi q(x) \,\du x + \omega_1 + \Omega_1 \right) + \ell_2 \left( \frac{1}{n} \right)
\end{equation}
and
\begin{equation} \label{eq:gamma}
  \gamma_n = \frac{\pi}{2} \left( n - \frac{\ind f + \ind F}{2} \right)^{2 \ind f} \left( 1 + \ell_2 \left( \frac{1}{n} \right) \right)
\end{equation}
respectively. Our aim in this subsection is to demonstrate that the same asymptotics also hold for all pairs $f$, $F \in \mathfrak{B}$, with $\omega_1$ and $\Omega_1$ defined as in~(\ref{eq:omega_Omega}). As a preliminary result, we show that the expression in parentheses on the right-hand side of~(\ref{eq:lambda}) is invariant under the transformation $\widehat{\mathbf{T}}$.

\begin{lemma} \label{lem:q_omega_Omega}
  If $(\widehat{q}, \widehat{f}, \widehat{F}) = \widehat{\mathbf{T}} (q, f, F)$ then
\begin{equation*}
  \frac{1}{2} \int_0^\pi q(x) \,\du x + \omega_1 + \Omega_1 = \frac{1}{2} \int_0^\pi \widehat{q}(x) \,\du x + \widehat{\omega}_1 + \widehat{\Omega}_1,
\end{equation*}
where $\widehat{\omega}_1$ and $\widehat{\Omega}_1$ are defined as in~(\ref{eq:omega_Omega}) with $f$ and $F$ replaced by $\widehat{f}$ and $\widehat{F}$ respectively.
\end{lemma}
\begin{proof}
It is easy to check that
\begin{equation*}
  \widehat{\omega}_1 - \omega_1 = \begin{cases} f(\lambda_0), & \ind f \ge 0, \\ - \frac{\ell_f + 1}{\pi}, & \ind f \le -1 \end{cases}
\end{equation*}
and similarly for the right endpoint (see the proof of~\cite[Lemma 4.1]{G17} for details). On the other hand, from~(\ref{eq:q_f_F_hat}) we have
\begin{equation*}
\begin{aligned}
  \frac{1}{2} \int_0^\pi \left( q(x) - \widehat{q}(x) \right) \,\du x = & \left. \left( \frac{\varphi'(x, \lambda_0)}{\varphi(x, \lambda_0)} - \frac{\ell_f + 1}{x} + \frac{\ell_F + 1}{\pi - x} \right) \right|_0^{\pi} \\
  = & \left. \left( \frac{\varphi'(x, \lambda_0)}{\varphi(x, \lambda_0)} + \frac{\ell_F + 1}{\pi - x} \right) \right|_{x=\pi} - \frac{\ell_F + 1}{\pi} \\
  & - \left. \left( \frac{\varphi'(x, \lambda_0)}{\varphi(x, \lambda_0)} - \frac{\ell_f + 1}{x} \right) \right|_{x=0} - \frac{\ell_f + 1}{\pi},
\end{aligned}
\end{equation*}
where the values at singular endpoints should be understood as limits. Now it only remains to observe that
\begin{equation*}
  \left. \left( \frac{\varphi'(x, \lambda_0)}{\varphi(x, \lambda_0)} - \frac{\ell_f + 1}{x} \right) \right|_{x=0} + \frac{\ell_f + 1}{\pi} = \omega_1 - \widehat{\omega}_1
\end{equation*}
and similarly for the right endpoint.
\end{proof}

Consider now the chain of problems~(\ref{eq:P_k}). The above lemma, Theorem \ref{thm:transformation}, and the obvious identities $\ind f^{(k)} - \ind f^{(k-1)} = \ind F^{(k)} - \ind F^{(k-1)} = 1$ show that if the asymptotics (\ref{eq:lambda}) and (\ref{eq:gamma}) hold for some problem of this chain then they also hold for the preceding one. Since the last problem $\mathscr{P}(q^{(K)}, f^{(K)}, F^{(K)})$ has no singular endpoints and thus its eigenvalues and norming constants satisfy (\ref{eq:lambda}) and (\ref{eq:gamma}), the same asymptotics also hold for all the other problems of the chain~(\ref{eq:P_k}), and in particular for the original problem $\mathscr{P}(q, f, F)$. We formulate this result as a theorem.

\begin{theorem} \label{thm:asymptotics}
  The eigenvalues and the norming constants of the problem $\mathscr{P}(q, f, F)$ satisfy the asymptotics (\ref{eq:lambda}) and (\ref{eq:gamma}).
\end{theorem}

\subsection{Oscillation of eigenfunctions} \label{ss:oscillation}

The Sturm oscillation theorem says that the number of zeros in $(0, \pi)$ of an eigenfunction of the Sturm--Liouville problem with constant boundary conditions equals the index of the corresponding eigenvalue. In the case of a boundary value problem $\mathscr{P}(q, f, F)$ with rational Herglotz--Nevanlinna functions $f$ and $F$, one needs to take into account the number of poles of these functions \cite[Theorem 4.4]{G17}. In this subsection, we generalize this result to the case of arbitrary $f$, $F \in \mathfrak{B}$.

To this end, we first prove that if $\varphi$ and $\widehat{\varphi}$ are the left regular solutions of two problems $\mathscr{P}(q, f, F)$ and $\mathscr{P}(\widehat{q}, \widehat{f}, \widehat{F})$ with $(\widehat{q}, \widehat{f}, \widehat{F}) = \widehat{\mathbf{T}} (q, f, F)$ then the difference between the number of zeros of $\varphi(x, \lambda_n)$ and $\widehat{\varphi}(x, \lambda_n)$ can be written as $1 + \boldsymbol{\Pi}_{\widehat{f}}(\lambda_n) - \boldsymbol{\Pi}_f(\lambda_n) + \boldsymbol{\Pi}_{\widehat{F}}(\lambda_n) - \boldsymbol{\Pi}_F(\lambda_n)$ for every $n \ge 1$. This is established by using the identities
\begin{equation*}
  \left( \widehat{\varphi}(x, \lambda_n) \varphi(x, \lambda_0) \right)' = \varphi(x, \lambda_n) \varphi(x, \lambda_0), \qquad \left( \frac{\varphi(x, \lambda_n)}{\varphi(x, \lambda_0)} \right)' = (\lambda_0 - \lambda_n) \frac{\widehat{\varphi}(x, \lambda_n)}{\varphi(x, \lambda_0)}
\end{equation*}
and showing that the function $\varphi(x, \lambda_n)$ has exactly one zero in each of the intervals $(0, x_1)$, $(x_1, x_2)$, $\ldots$, $(x_N, \pi)$, where $x_1$, $\ldots$, $x_N$ are the zeros of the function $\widehat{\varphi}(x, \lambda_n)$ in $(0, \pi)$. More details can be found in \cite[Lemma 4.3]{G17} when $f$ and $F$ are rational Herglotz--Nevanlinna functions, and the case of singular endpoints is even easier as the eigenfunctions vanish at singular endpoints.

Consider again the problems~(\ref{eq:P_k}). The $n$-th eigenfunction of $\mathscr{P}(q, f, F)$ corresponds to the $(n+K)$-th eigenfunction of $\mathscr{P}(q^{(K)}, f^{(K)}, F^{(K)})$ and the latter one has $n + K - \boldsymbol{\Pi}_{f^{(K)}}(\lambda_n) - \boldsymbol{\Pi}_{F^{(K)}}(\lambda_n)$ zeros in $(0, \pi)$.
Therefore, applying the result of the preceding paragraph successively to the problems~(\ref{eq:P_k}), we obtain the following result.

\begin{theorem} \label{thm:oscillation}
  An eigenfunction of the problem $\mathscr{P}(q, f, F)$ corresponding to the eigenvalue $\lambda_n$ has exactly $n - \boldsymbol{\Pi}_f(\lambda_n) - \boldsymbol{\Pi}_F(\lambda_n)$ zeros in $(0, \pi)$.
\end{theorem}

\subsection{Regularized trace formulas} \label{ss:trace}

As can be seen from (\ref{eq:Riccati}), (\ref{eq:q_f_F_hat}), and (\ref{eq:q_f_F_tilde}), our transformations $\widehat{\mathbf{T}}$ and $\widetilde{\mathbf{T}}$ preserve the smoothness of the potential. This implies, in particular, that the eigenvalues of a problem $\mathscr{P}(q, f, F)$ with $q \in \mathscr{W}_2^1[0, \pi]$ have the asymptotics
\begin{equation*}
  \sqrt{\lambda_n} = n - a + \frac{b}{n - a} + \ell_2 \left( \frac{1}{n^2} \right),
\end{equation*}
where
\begin{equation*}
  a := \frac{\ind f + \ind F}{2}, \qquad b := \frac{1}{\pi} \left( \frac{1}{2} \int_0^\pi q(x) \,\du x + \omega_1 + \Omega_1 \right).
\end{equation*}
Therefore, \emph{the first regularized trace} of the problem $\mathscr{P}(q, f, F)$, which we defined in \cite{G17} as
\begin{equation*}
  \Trace(q, f, F) := \sum_{n < a} \lambda_n + \sum_{n = a} (\lambda_n - b) + \sum_{n > a} \left( \lambda_n - (n - a)^2 - 2b \right),
\end{equation*}
converges. In the case of a boundary value problem $\mathscr{P}(q, f, F)$ with rational Herglotz--Nevanlinna functions $f$ and $F$, the sum of this series can be calculated by the formula~\cite[Theorem 4.6]{G17}
\begin{equation*}
  \Trace(q, f, F) = \frac{(-1)^{\ind f} q(0)}{4} + \frac{(-1)^{\ind F} q(\pi)}{4} - \frac{\omega_1^2}{2} - \frac{\Omega_1^2}{2} - \omega_2 - \Omega_2.
\end{equation*}
In order to also include inverse square singularities, the expression on the right-hand side needs to be generalized slightly.

\begin{theorem} \label{thm:trace}
The following identity holds:
\begin{equation*}
\begin{aligned}
  \Trace(q, f, F) = & (-1)^{\ell_f + \ind f} \frac{2 \ell_f + 1}{4} \left( q(0) + \frac{\ell_F (\ell_F + 1)}{\pi^2} \right) \\
  & + (-1)^{\ell_F + \ind F} \frac{2 \ell_F + 1}{4} \left( q(\pi) + \frac{\ell_f (\ell_f + 1)}{\pi^2} \right) \\
  & - \frac{\omega_1^2}{2} - \frac{\Omega_1^2}{2} - \omega_2 - \Omega_2 - \frac{(a^2 + a + 6b) (2 a + 1)}{6} \mathbf{1}_{(-\infty,-1]}(a).
\end{aligned}
\end{equation*}
\end{theorem}

For rational Herglotz--Nevanlinna functions $f$ and $F$ we have $\ell_f = \ell_F = -1$ and $a \ge 0$, and thus the two expressions coincide. The proof of the theorem follows by reducing the problem $\mathscr{P}(q, f, F)$ to the one with $\min \{ \ind f, \ind F \} \ge -1$, as in the preceding two subsections, and using Lemma \ref{lem:q_omega_Omega} and the following lemma.

\begin{lemma} \label{lem:trace}
Let $(\widehat{q}, \widehat{f}, \widehat{F}) := \widehat{\mathbf{T}} (q, f, F)$, and let $\widehat{\omega}_1$ and $\widehat{\omega}_2$ be defined by~(\ref{eq:omega_Omega}) with $f$ and $F$ replaced by $\widehat{f}$ and $\widehat{F}$ respectively. We have
\begin{multline*}
  (-1)^{\ell_{\widehat{f}} + \ind \widehat{f}} (2 \ell_{\widehat{f}} + 1)\left( \widehat{q}(0) + \frac{\ell_{\widehat{F}} (\ell_{\widehat{F}} + 1)}{\pi^2} \right) - 2 \widehat{\omega}_1^2 - 4 \widehat{\omega}_2 \\
  = (-1)^{\ell_f + \ind f} (2 \ell_f + 1) \left( q(0) + \frac{\ell_F (\ell_F + 1)}{\pi^2} \right) - 2 \omega_1^2 - 4 \omega_2 - 2 \lambda_0
\end{multline*}
and similarly for the right endpoint. The regularized traces of these two problems are related by the formula
\begin{equation*}
  \Trace(\widehat{q}, \widehat{f}, \widehat{F}) = \Trace(q, f, F) - \lambda_0 + \left( a^2 + 2 b \right) \mathbf{1}_{(-\infty,0)}(a) + b \mathbf{1}_{\{0\}}(a).
\end{equation*}
\end{lemma}
\begin{proof}
From (\ref{eq:Riccati}) and (\ref{eq:q_f_F_hat}) we obtain
\begin{equation*}
  \frac{\widehat{q}(0) + q(0)}{2} = \left. \left( \left( \frac{\varphi'(x, \lambda_0)}{\varphi(x, \lambda_0)} \right)^2 - \frac{(\ell_f + 1)^2}{x^2} \right) \right|_{x=0} - \frac{(\ell_F + 1)^2}{\pi^2} + \lambda_0,
\end{equation*}
which, using (\ref{eq:fine_asymptotics}), gives
\begin{equation*}
  (2 \ell_f + 3) \left( \widehat{q}(0) + \frac{(\ell_F + 1) (\ell_F + 2)}{\pi^2} \right) = (2 \ell_f + 1) \left( q(0) + \frac{\ell_F (\ell_F + 1)}{\pi^2} \right) + 2 f^2(\lambda_0) + 2 \lambda_0
\end{equation*}
with the convention that the next to the last term on the right-hand side is omitted whenever $\ind f \le -1$. This together with the identities
\begin{equation*}
  (2 \ell_f + 3) = (-1)^{\ell_{\widehat{f}} - \ell_f + \ind \widehat{f} - \ind f} (2 \ell_{\widehat{f}} + 1), \qquad (\ell_F + 1) (\ell_F + 2) = \ell_{\widehat{F}} (\ell_{\widehat{F}} + 1),
\end{equation*}
and
\begin{equation*}
  \widehat{\omega}_1^2 + 2 \widehat{\omega}_2 - \omega_1^2 - 2 \omega_2 = \lambda_0 + (-1)^{\ell_f + \ind f} \left( f^2(\lambda_0) + \lambda_0 \right)
\end{equation*}
proves the first identity of the lemma. The case of the right endpoint is proved similarly. The relation between $\Trace(\widehat{q}, \widehat{f}, \widehat{F})$ and $\Trace(q, f, F)$ is straightforward from the definition of the trace.
\end{proof}

\subsection{Inverse problem by eigenvalues and norming constants} \label{ss:byspectraldata}

We now turn to inverse eigenvalue problems for boundary value problems of the form $\mathscr{P}(q, f, F)$. The purpose of this subsection is to answer the following question: What conditions must two sequences of real numbers $\{ \lambda_n \}_{n \ge 0}$ and $\{ \gamma_n \}_{n \ge 0}$ satisfy in order to be the eigenvalues and the norming constants of a boundary value problem of the form $\mathscr{P}(q, f, F)$?

Obviously, the elements of the former sequence must be pairwise distinct and the elements of the latter one must all be positive:
\begin{equation} \label{eq:increasing}
  \lambda_0 < \lambda_1 < \lambda_2 < \ldots, \qquad \gamma_n > 0, \quad n \ge 0.
\end{equation}
Moreover, Theorem~\ref{thm:asymptotics} shows that these two sequences should also satisfy the asymptotics
\begin{equation} \label{eq:asymptotics}
\begin{aligned}
  \sqrt{\lambda_n} &= n - \frac{M + N}{2} + \frac{\sigma}{\pi n} + \ell_2 \left( \frac{1}{n} \right), \\
  \gamma_n &= \frac{\pi}{2} \left( n - \frac{M + N}{2} \right)^{2M} \left( 1 + \ell_2 \left( \frac{1}{n} \right) \right)
\end{aligned}
\end{equation}
for some real $\sigma$ and integers $M$, $N$. We are now going to prove that these necessary conditions are also sufficient, i.e., for any sequences of real numbers $\{ \lambda_n \}_{n \ge 0}$ and $\{ \gamma_n \}_{n \ge 0}$ satisfying these conditions, there exists a unique boundary value problem $\mathscr{P}(q, f, F)$ having these sequences as its eigenvalues and norming constants. The idea is to reduce the general case to the case $M$, $N \ge -1$ for which the result was already proved in~\cite[Theorem 4.7]{G17}.

\begin{theorem} \label{thm:by_spectral_data}
  Let $\{ \lambda_n \}_{n \ge 0}$ and $\{ \gamma_n \}_{n \ge 0}$ be sequences of real numbers satisfying the conditions (\ref{eq:increasing}) and (\ref{eq:asymptotics}). Then there exists a unique boundary value problem $\mathscr{P}(q, f, F)$ having the eigenvalues $\{ \lambda_n \}_{n \ge 0}$ and the norming constants $\{ \gamma_n \}_{n \ge 0}$.
\end{theorem}
\begin{proof}
We only need to consider the case $K := -1 - \min \{ M, N \} \ge 1$. Let $\lambda_{-K} < \ldots < \lambda_{-1}$ be arbitrary numbers strictly less than $\lambda_0$ and let $\gamma_{-K}$, $\ldots$, $\gamma_{-1}$ be arbitrary positive numbers. With Theorem~\ref{thm:inverse_transformation} in mind, we consider the sequences $\{ \widetilde{\lambda}_n \}_{n \ge 0}$ and $\{ \widetilde{\gamma}_n \}_{n \ge 0}$ defined by
\begin{equation*}
  \widetilde{\lambda}_n := \lambda_{n-K}, \qquad \widetilde{\gamma}_n := \gamma_{n-K} \prod_{m=0}^{\min\{ n, K \} - 1} (\lambda_{n-K} - \lambda_{m-K}).
\end{equation*}
They satisfy the conditions (\ref{eq:increasing}) and (\ref{eq:asymptotics}) with $M$ and $N$ replaced by $M + K \ge -1$ and $N + K \ge -1$ respectively. Therefore there exists a boundary value problem $\mathscr{P}(\widetilde{q}, \widetilde{f}, \widetilde{F})$ having the eigenvalues $\{ \widetilde{\lambda}_n \}_{n \ge 0}$ and the norming constants $\{ \widetilde{\gamma}_n \}_{n \ge 0}$. Now it only remains to observe from Theorem~\ref{thm:transformation} that the original sequences $\{ \lambda_n \}_{n \ge 0}$ and $\{ \gamma_n \}_{n \ge 0}$ are the eigenvalues and the norming constants of the boundary value problem $\mathscr{P}(q, f, F)$ with $(q, f, F) := \widehat{\mathbf{T}}^K (\widetilde{q}, \widetilde{f}, \widetilde{F})$. The uniqueness follows similarly from Theorems~\ref{thm:transformation} and \ref{thm:inverse}.
\end{proof}

\subsection{Inverse problems by one spectrum} \label{ss:byonespectrum}

It is well known for a boundary value problem $\mathscr{P}(q, f, F)$ with constant $f$ and $F$ that the knowledge of the norming constants can be replaced by some other information, like a symmetry assumption on this boundary value problem, or the knowledge of the potential on half the interval and the corresponding boundary constant. We generalized these results to problems $\mathscr{P}(q, f, F)$ with rational Herglotz--Nevanlinna functions $f$ and $F$ in~\cite{G17}. These results can also be extended to the case of arbitrary $f$, $F \in \mathfrak{B}$.

We start with \emph{symmetric} boundary value problems, i.e., those of the form $\mathscr{P}(q, f, f)$ with $q(x) = q(\pi - x)$. Obviously, the eigenvalues of a symmetric problem satisfy asymptotics of the form~(\ref{eq:asymptotics}) with $M = N$. Conversely, if $\{ \lambda_n \}_{n \ge 0}$ is a sequence of real numbers satisfying the asymptotics
\begin{equation} \label{eq:symmetric_asymptotics}
  \sqrt{\lambda_n} = n - M + \frac{\sigma}{\pi n} + \ell_2 \left( \frac{1}{n} \right),
\end{equation}
for some real $\sigma$ and integer $M$, then there exists a unique symmetric problem $\mathscr{P}(\widetilde{q}, \widetilde{f}, \widetilde{f})$ having the eigenvalues $\{ \widetilde{\lambda}_n \}_{n \ge 0}$ constructed as in the preceding subsection. The problem $\mathscr{P}(q, f, F)$ with $(q, f, F) := \widehat{\mathbf{T}}^K (\widetilde{q}, \widetilde{f}, \widetilde{f})$ is also symmetric and its uniqueness can be established in a similar way. Hence we have the following result. Alternatively, one could first prove that a problem is symmetric if and only if $\beta_n = (-1)^n$ for all $n \ge 0$ and then use (\ref{eq:chi_beta_gamma}), as in \cite[Subsection 4.5]{G17}.

\begin{theorem} \label{thm:symmetric}
  Let $\{ \lambda_n \}_{n \ge 0}$ be a strictly increasing sequence of real numbers satisfying the asymptotics (\ref{eq:symmetric_asymptotics}) for some real $\sigma$ and integer $M$. Then there exists a unique symmetric boundary value problem $\mathscr{P}(q, f, f)$ having the spectrum $\{ \lambda_n \}_{n \ge 0}$.
\end{theorem}

Let us now turn to Hochstadt--Lieberman-type results. For the sake of variety, here we transform every problem $\mathscr{P}(q, f, F)$ to a problem with singular endpoints rather than with boundary conditions dependent on the eigenvalue parameter. It follows from a result of Eckhardt and Teschl~\cite[Theorem 5.4]{ET13} that a problem $\mathscr{P}(q, f, F)$ with $\max \{ \ind f, \ind F \} \le -1$ is uniquely determined by its spectrum, the coefficient $f$, and the values of $q$ on $(0, \pi / 2 + \varepsilon)$ for some $\varepsilon > 0$, where one can also take $\varepsilon = 0$ if $\ind f \ge \ind F$.

So, let $\mathscr{P}(q, f, F)$ and $\mathscr{P}(\widetilde{q}, f, \widetilde{F})$ be two problems with $q(x) = \widetilde{q}(x)$ a.e. on $(0, \pi / 2 + \varepsilon)$ and the same eigenvalues $\lambda_n$, $n \ge 0$. The isospectrality of the boundary value problems $\mathscr{P}(q, f, F)$ and $\mathscr{P}(\widetilde{q}, f, \widetilde{F})$ and the asymptotics of their eigenvalues imply that $\ind F = \ind \widetilde{F}$. We assume without loss of generality that $K := 1 + \max \{ \ind f, \ind F \} \ge 1$, since the other case is already contained in the above-mentioned result. We set $(q^{(0)}, f^{(0)}, F^{(0)}) := (q, f, F)$ and $(\widetilde{q}^{(0)}, \widetilde{f}^{(0)}, \widetilde{F}^{(0)}) := (\widetilde{q}, f, \widetilde{F})$, and define the chains of problems $\mathscr{P}(q^{(k)}, f^{(k)}, F^{(k)})$ and $\mathscr{P}(\widetilde{q}^{(k)}, \widetilde{f}^{(k)}, \widetilde{F}^{(k)})$ inductively by $(q^{(k)}, f^{(k)}, F^{(k)}) := \widehat{\mathbf{T}} (q^{(k-1)}, f^{(k-1)}, F^{(k-1)})$ and $(\widetilde{q}^{(k)}, \widetilde{f}^{(k)}, \widetilde{F}^{(k)}) := \widehat{\mathbf{T}} (\widetilde{q}^{(k-1)}, \widetilde{f}^{(k-1)}, \widetilde{F}^{(k-1)})$ for $k = 1$, $2$, $\ldots$, $K$. Then obviously $f^{(k)} = \widetilde{f}^{(k)}$ and $\ind F^{(k)} = \ind \widetilde{F}^{(k)}$ for each $k$.

Denote the left regular solutions of the boundary value problems $\mathscr{P}(q^{(k)}, f^{(k)}, F^{(k)})$ and $\mathscr{P}(\widetilde{q}^{(k)}, \widetilde{f}^{(k)}, \widetilde{F}^{(k)})$ by $\varphi^{(k)}$ and $\widetilde{\varphi}^{(k)}$ respectively. Then using the definitions of the left regular solution and the transformation $\widehat{\mathbf{T}}$, for each $k = 0$, $1$, $\ldots$, $K-1$ we successively obtain $\varphi^{(k)}(x, \lambda_k) = \widetilde{\varphi}^{(k)}(x, \lambda_k)$ on $(0, \pi / 2 + \varepsilon)$ and thus $q^{(k+1)}(x) = \widetilde{q}^{(k+1)}(x)$ a.e. on $(0, \pi / 2 + \varepsilon)$. In particular, we have $f^{(K)} = \widetilde{f}^{(K)}$ and $q^{(K)}(x) = \widetilde{q}^{(K)}(x)$ a.e. on $(0, \pi / 2 + \varepsilon)$, and hence, by the result mentioned above, $(q^{(K)}, f^{(K)}, F^{(K)}) = (\widetilde{q}^{(K)}, \widetilde{f}^{(K)}, \widetilde{F}^{(K)})$.

We now observe that both $1 / \varphi^{(K-1)}(x, \lambda_{K-1})$ and $1 / \widetilde{\varphi}^{(K-1)}(x, \lambda_{K-1})$ coincide on the left half of the interval $(0, \pi)$ and satisfy the same differential equation on all of this interval. This yields $q^{(K-1)}(x) = \widetilde{q}^{(K-1)}(x)$ a.e. on $(0, \pi)$ and $F^{(K-1)} = \widetilde{F}^{(K-1)}$. Repeating this argument $K - 1$ more times concludes the proof of the following result.

\begin{theorem} \label{thm:partial}
  Let $\{ \lambda_n \}_{n \ge 0}$ and $\{ \widetilde{\lambda}_n \}_{n \ge 0}$ denote the eigenvalues of the problems $\mathscr{P}(q, f, F)$ and $\mathscr{P}(\widetilde{q}, f, \widetilde{F})$ respectively. If $q(x) = \widetilde{q}(x)$ a.e. on $(0, \pi / 2 + \varepsilon)$ for some $\varepsilon > 0$ and $\lambda_n = \widetilde{\lambda}_n$ for all $n \ge 0$, then $(q, f, F) = (\widetilde{q}, f, \widetilde{F})$. Moreover, one can also take $\varepsilon = 0$ if $\ind f \ge \ind F$.
\end{theorem}


\begin{thebibliography}{99}

\bibitem{AHM07} S. Albeverio, R. Hryniv, and Ya. Mykytyuk,
\emph{Inverse spectral problems for Bessel operators},
J. Differential Equations \textbf{241} (2007), no. 1, 130--159.

\bibitem{AEZ88} F. V. Atkinson, W. N. Everitt, and A. Zettl,
\emph{Regularization of a Sturm--Liouville problem with an interior singularity using quasiderivatives},
Differential Integral Equations \textbf{1} (1988), no. 2, 213--221.

\bibitem{BBW02a} P. A. Binding, P. J. Browne and B. A. Watson,
\emph{Sturm--Liouville problems with boundary conditions rationally dependent on the eigenparameter. I},
Proc. Edinb. Math. Soc. (2) \textbf{45} (2002), no. 3, 631--645.

\bibitem{BBW02b} P. A. Binding, P. J. Browne and B. A. Watson,
\emph{Sturm--Liouville problems with boundary conditions rationally dependent on the eigenparameter. II},
J. Comput. Appl. Math. \textbf{148} (2002), no. 1, 147--168.

\bibitem{C93} R. Carlson,
\emph{Inverse spectral theory for some singular Sturm--Liouville problems},
J. Differential Equations \textbf{106} (1993), no. 1, 121--140.

\bibitem{C42} R. V. Churchill,
\emph{Expansions in series of non-orthogonal functions},
Bull. Amer. Math. Soc. \textbf{48} (1942), 143--149.

\bibitem{C55} M. M. Crum,
\emph{Associated Sturm--Liouville systems},
Quart. J. Math. Oxford Ser. (2) \textbf{6} (1955), 121--127.
\href{https://arxiv.org/abs/physics/9908019}{arXiv:physics/9908019}

\bibitem{DL17} A. Dijksma and H. Langer,
\emph{Finite-dimensional self-adjoint extensions of a symmetric operator with finite defect and their compressions},
Advances in complex analysis and operator theory, Birkh\"{a}user/Springer, Cham, 2017, pp. 135--163.

\bibitem{DL18} A. Dijksma and H. Langer,
\emph{Compressions of self-adjoint extensions of a symmetric operator and M. G. Krein's resolvent formula},
Integral Equations Operator Theory \textbf{90} (2018), no. 4, Art. 41, 30 pp.

\bibitem{D67} P. A. M. Dirac,
\emph{The principles of quantum mechanics},
Clarendon, Oxford, 1967.

\bibitem{ET13} J. Eckhardt and G. Teschl,
\emph{Uniqueness results for one-dimensional Schr\"{o}dinger operators with purely discrete spectra},
Trans. Amer. Math. Soc. \textbf{365} (2013), no. 7, 3923--3942.
\href{https://arxiv.org/abs/1110.2453}{arXiv:1110.2453}

\bibitem{F77} C. T. Fulton,
\emph{Two-point boundary value problems with eigenvalue parameter contained in the boundary conditions},
Proc. Roy. Soc. Edinburgh Sect. A \textbf{77} (1977), no. 3-4, 293--308.

\bibitem{F80} C. T. Fulton,
\emph{Singular eigenvalue problems with eigenvalue parameter contained in the boundary conditions},
Proc. Roy. Soc. Edinburgh Sect. A \textbf{87} (1980), no. 1-2, 1--34.

\bibitem{G65} M. G. Gasymov,
\emph{Determination of a Sturm--Liouville equation with a singularity by two spectra} (Russian),
Dokl. Akad. Nauk SSSR \textbf{161} (1965), 274--276; English transl. in Soviet Math. Dokl. \textbf{6} (1965), 396--399.

\bibitem{GT96} F. Gesztesy and G. Teschl,
\emph{On the double commutation method},
Proc. Amer. Math. Soc. \textbf{124} (1996), no. 6, 1831--1840.

\bibitem{G17} N. J. Guliyev,
\emph{Essentially isospectral transformations and their applications},
Ann. Mat. Pura Appl. (4), to appear.
\href{https://arxiv.org/abs/1708.07497}{arXiv:1708.07497}

\bibitem{G18a} N. J. Guliyev,
\emph{On two-spectra inverse problems},
submitted.
\href{https://arxiv.org/abs/1803.02567}{arXiv:1803.02567}

\bibitem{G18b} N. J. Guliyev,
\emph{On extensions of symmetric operators},
Oper. Matrices, to appear.
\href{https://arxiv.org/abs/1807.11865}{arXiv:1807.11865}

\bibitem{G19a} N. J. Guliyev,
\emph{Schr\"{o}dinger operators with distributional potentials and boundary conditions dependent on the eigenvalue parameter},
J. Math. Phys. \textbf{60} (2019), no. 6, 063501, 23 pp.
\href{https://arxiv.org/abs/1806.10459}{arXiv:1806.10459}

\bibitem{G19b} N. J. Guliyev,
\emph{A Riesz basis criterion for Schr\"{o}dinger operators with boundary conditions dependent on the eigenvalue parameter},
Anal. Math. Phys. \textbf{10} (2020), no. 1, 2, 8 pp.
\href{https://arxiv.org/abs/1905.07952}{arXiv:1905.07952}

\bibitem{G19c} N. J. Guliyev,
\emph{Spectral identities for Schr\"{o}dinger operators},
submitted.
\href{https://arxiv.org/abs/1910.05812}{arXiv:1910.05812}

\bibitem{KST10} A. Kostenko, A. Sakhnovich, and G. Teschl,
\emph{Inverse eigenvalue problems for perturbed spherical Schr\"{o}dinger operators},
Inverse Problems \textbf{26} (2010), no. 10, 105013, 14 pp.
\href{https://arxiv.org/abs/1004.4175}{arXiv:1004.4175}

\bibitem{KST12} A. Kostenko, A. Sakhnovich, and G. Teschl,
\emph{Commutation methods for Schr\"{o}dinger operators with strongly singular potentials},
Math. Nachr. \textbf{285} (2012), no. 4, 392--410.
\href{https://arxiv.org/abs/1010.4902}{arXiv:1010.4902}

\bibitem{K57} M. G. Krein,
\emph{On a continual analogue of a Christoffel formula from the theory of orthogonal polynomials} (Russian),
Dokl. Akad. Nauk SSSR \textbf{113} (1957), 970--973.

\bibitem{M77} V. A. Marchenko,
\emph{Sturm--Liouville operators and applications} (Russian),
Naukova Dumka, Kiev, 1977;
English transl.: AMS Chelsea Publishing, Providence, RI, 2011.

\bibitem{PT87} J. P\"{o}schel and E. Trubowitz,
\emph{Inverse spectral theory},
Academic Press, Inc., Boston, MA, 1987.

\bibitem{S15} B. Simon,
\emph{Basic complex analysis. A Comprehensive Course in Analysis, Part 2A},
American Mathematical Society, Providence, RI, 2015.

\end{thebibliography}
\end{document}